%% file: main.tex
\documentclass[11pt]{article}%
\usepackage{amsmath,amsthm,amssymb,amsfonts}
\usepackage[colorlinks]{hyperref}
\hypersetup{linkcolor=blue, citecolor=blue}
\usepackage{cleveref}
\usepackage{graphicx}
\usepackage{fullpage}
\usepackage{tikz}
\usetikzlibrary{quantikz2}
\usepackage{appendix}

\usepackage{complexity}
\newcommand{\QTC}{\ComplexityFont{QTC}}
\newcommand{\BQTC}{\ComplexityFont{BQTC}}
\newcommand{\BQAC}{\ComplexityFont{BQAC}}

\newtheorem{theorem}{Theorem}

\newtheorem{corollary}[theorem]{Corollary}
\newtheorem{definition}[theorem]{Definition}
\newtheorem{lemma}[theorem]{Lemma}
\newtheorem{proposition}[theorem]{Proposition}

\newcommand{\AND}{\mathsf{AND}}

\newcommand{\XOR}{\mathsf{XOR}}
\newcommand{\CNOT}{\mathsf{CNOT}}
\newcommand{\CZ}{\mathsf{CZ}}
\renewcommand{\MOD}{\mathsf{MOD}}

\newcommand{\Fp}{\mathbb{F}_p}
\newcommand{\Maj}{\operatorname{Maj}}

\newcommand{\Fanout}{\mathsf F} 
\renewcommand{\M}{\mathsf M}

\newcommand{\Threshold}{\mathsf{Th}} 
 
\newcommand{\Parity}{\mathsf P} 

\begin{document}

\title{Quantum Threshold is Powerful}
\author{
Daniel Grier\thanks{UCSD. \ Email: \texttt{dgrier@ucsd.edu}}
\and 
Jackson Morris\thanks{UCSD. \ Email: \texttt{jrm035@ucsd.edu}}
}
\date{}
\maketitle

\begin{abstract}
In 2005, H{\o}yer and {\v S}palek showed that constant-depth quantum circuits augmented with multi-qubit Fanout gates are quite powerful, able to compute a wide variety of Boolean functions as well as the quantum Fourier transform. They also asked what other multi-qubit gates could rival Fanout in terms of computational power, and suggested that the quantum Threshold gate might be one such candidate. Threshold is the gate that indicates if the Hamming weight of a classical basis state input is greater than some target value. 

We prove that Threshold is indeed powerful---there are polynomial-size constant-depth quantum circuits with Threshold gates that compute Fanout to high fidelity. Our proof is a generalization of a proof by Rosenthal that exponential-size constant-depth circuits with generalized Toffoli gates can compute Fanout. Our construction reveals that other quantum gates able to ``weakly approximate'' Parity can also be used as substitutes for Fanout.

\end{abstract}

\section{Introduction\label{INTRO}}

To what extent are large multi-qubit gates useful for quantum computation? On the one hand, it is well-known that every multi-qubit gate can be decomposed into a circuit of simpler $1$- and $2$-qubit gates. On the other hand, this decomposition may introduce large overheads both in terms of gate count and circuit depth. Given that some multi-qubit gates might be experimentally feasible \cite{saffman2009efficient, levine2019parallel, molmer1999multiparticle}, it's natural to ask what kinds of computational powers they unlock. 

Specifically, we focus on the power of these large multi-qubit gates in constant depth. Such shallow circuits are experimentally appealing due to the possibility for less decoherence. Moreover, even shallow quantum circuits with $1$- and $2$-qubit gates are known to be surprisingly powerful, exhibiting quantum advantage in a variety of settings \cite{bravyi2018quantum, grier2020interactive, terhal2002adaptive, haferkamp2020closing}. Given the inherent complexity of simulating such circuits, there is the exciting possibility that augmenting these circuit models with large multi-qubit gates might lead to constant-depth implementations of practical quantum algorithms.

Much of the excitement about such circuit models is driven by a single gate---the multi-qubit Fanout gate---which is the quantum operation that copies \emph{classical} information:
\[
\Fanout_n \ket{b, x_1, \ldots, x_n} := \ket{b, x_1 \oplus b, \ldots, x_n \oplus b}
\]
for all $b, x_1, \ldots, x_n \in \{0,1\}$. 

This seemingly innocuous gate (which is included for free in almost every classical circuit model) turns out to be quite powerful. For starters, it is locally equivalent via conjugation by Hadamard gates to the quantum Parity gate \cite{moore:1999},
\[
\Parity_n \ket{b, x_1, \ldots, x_n} := \ket{b \oplus (x_1 \oplus \cdots \oplus x_n), x_1, \ldots, x_n},
\]
which is a duality that has no classical counterpart \cite{ajtai_parity}. Moreover, there are constant-depth quantum circuits with Fanout (and arbitrary single-qubit gates) for a wide variety of other symmetric Boolean operations such as And/Or and Majority \cite{hs_fanout, tt_fanout}. Perhaps most impressively, constant-depth quantum circuits with Fanout gates can factor integers with polynomial-time classical post-processing~\cite{hs_fanout}. 

Given the centrality of Fanout to the story of low-depth circuits with multi-qubit gates, there has been significant work in trying to understand if other multi-qubit gates are similarly powerful. Most notably, it is widely believed that the multi-qubit generalization of the Toffoli gate is fundamentally less powerful than the Fanout gate in constant depth, and there is long line of work giving evidence that these generalized Toffoli gates cannot compute Fanout \cite{bera_lb, lb_fanout, rosenthal, depth2, pauli_spec, anshu_dong_ou_yao:2024_QAC0}. In some sense, all of these results are grappling with a fundamental tension in the study of these low-depth circuit models---the high entanglement in the states produced by these circuits is an obstacle to proving lower bounds, but it is simultaneously unclear how one could leverage this complexity to implement a useful quantum algorithm. 

There is a surprising dearth of low-depth circuit models with multi-qubit gates that are as powerful as Fanout. One natural\footnote{This candidate looks considerably more natural after considering the analogous landscape of \emph{classical} circuits, which we discuss in \Cref{sec:classical_comparison}.} candidate for a gate that could be as powerful as Fanout is the quantum Threshold gate, a multi-qubit gate parameterized by some value $k \in \mathbb N$:
\[
    \Threshold^n_k \ket{b, x_1, \ldots, x_n} := \ket{b \oplus \mathbb I_{|x| \ge k}, x_1, \ldots, x_n}
\]
where $\mathbb I_{|x| \ge k}$ indicates if the Hamming weight of the input bit string $x = x_1 \cdots x_n \in \{0,1\}^n$ is at least the target value $k$. In fact, Høyer and Špalek asked almost 20 years ago about the power of Threshold in constant depth \cite{hs_fanout}: ``Can we simulate unbounded fan-out in constant depth using unbounded fan-in gates, e.g. $\operatorname{threshold}[t]$ or $\operatorname{exact}[t]$?" This question was reiterated more pointedly by Takahashi and Tani in 2011 \cite{tt_fanout}: ``Does there exist a fundamental gate that is as powerful as an unbounded fan-out gate?" 

We directly answer both of these questions in the affirmative by giving explicit constructions for Fanout using quantum Threshold gates:
\begin{theorem}
\label{thm:intro_main}
    There are poly-size constant-depth quantum circuits consisting of Threshold gates and arbitrary single-qubit gates that compute Fanout with high fidelity. Formally, $\BQTC^0 = \BQNC^0_{wf}$.
\end{theorem}

The construction from this theorem actually reveals a number of other gates that are as fundamentally powerful as the Fanout gate. As it turns out, the salient feature of Threshold for our purposes is that it can be used to construct a sort of ``weak'' Parity gate---a gate that only acts non-trivially on inputs of the same parity. 

Based on this idea, we introduce a class of multi-qubit phase gates that exhibit a generalization of this behavior. Formally, these gates are defined with respect to a set $S \subset \{0,1\}^n$ in the following way:
\[
    U_S \ket{x_1, \ldots, x_n} := (-1)^{\mathbb I_{x \in S}}\ket{x_1, \ldots, x_n}.
\]
Crucially, we restrict our attention to ``parity-restricted" sets $S$, that is, sets where all elements have the same parity (i.e., $x, y \in S \implies |x| \equiv |y| \pmod 2$). We show that these weak parity gates can be bootstrapped in constant depth into true Parity gates (which, recall, are locally-equivalent to Fanout) albeit with the help of a few generalized Toffoli gates:
\begin{theorem}
\label{thm:intro_general}
    Let $\{S_n\}_n$ be a family of parity-restricted sets with size $|S_n| = \Theta(2^n/\poly(n))$. There are poly-size constant-depth quantum circuits consisting of $U_{S_n}$ gates, generalized Toffoli gates, and arbitrary single-qubit gates that compute Fanout with high fidelity.
\end{theorem}
Since it is widely believed that multi-qubit Toffoli gates are not themselves sufficient to implement Fanout, the power of this construction likely derives from the weak parity gates. In fact, the reason these Toffoli gates were not required for \Cref{thm:intro_main} is due to the fact that Threshold can directly simulate Toffoli. In that vein, we also give conditions under which the $U_{S_n}$ gates alone suffice to simulate Parity; namely, when $|S_n| \ge 2^{n-O(1)}$ or $|S_n| \le 2^{(1-\epsilon)n}$. Though, the later condition will result in circuits of super-polynomial size.

While it has long been thought that Fanout/Parity gates were morally equivalent to other quantum modular arithmetic gates, those constructions seem to \emph{also} require these generalized Toffoli gates \cite{qacc}. By a careful inspection of the original construction presented in \cite{qacc} we find that generalized Toffoli gates are in fact not necessary. Formally, the quantum Mod-$p$ gates is defined as
\[
    \MOD_p^n \ket{b, x_1, \ldots, x_n} := \ket{b \oplus \Mod_p^n(x), x_1, \ldots, x_n},
\]
where $\Mod_p^n(x)$ is $1$ when $p$ divides the Hamming weight of $x = x_1, \cdots, x_n \in \{0,1\}^n$. For example, the Mod-$2$ gate is essentially the Parity gate (up to a single-qubit $X$ gate). It is implicit in \cite{qacc} that Fanout can be computed by a circuit consisting of Mod-$p$ gates and one- and two-qubit gates, yielding $\QNC^0_{wf} = \QNC^0[2] \subseteq \QNC^0[q]$ for all $q \geq 2$, but not necessarily that $\QNC^0[p] = \QNC^0[q]$ for distinct $p$ and $q$. The result they make explicit is that when Toffoli gates are allowed, any Mod-$q$ gate can be obtained using any other Mod-$p$ gate (by first implementing Fanout with Mod-$q$ gates and then computing Mod-$p$ with Fanout and generalized Toffoli gates). Concretely; $\QAC^0[p] = \QAC^0[q]$ for $p, q \geq 2$. Only later was it shown that generalized Toffoli gates can be implemented using Fanout and single- and two-qubit gates i.e. that $\QNC^0_{wf} = \QAC^0_{wf}$ \cite{hs_fanout, tt_fanout}. In light of these results we observe the following:
\begin{theorem}
For all $p, q \geq 2$, there are poly-size constant-depth quantum circuits consisting of Mod-$p$ gates and  single-qubit gates that compute the Mod-$q$ operation. Formally, $\QNC^0[p] = \QNC^0[q]$.
\end{theorem}

\subsection{Comparison to the classical setting}
\label{sec:classical_comparison}

Our focus on shallow circuits draws considerable inspiration from the analogous study of classical constant-depth circuit classes with large fan-in gates, which has been hugely influential in classical complexity theory. For instance, initial work in Boolean circuits saw the development of techniques for proving unconditional lower bounds such as random restrictions \cite{ajtai_parity, fss_ac0, yao_ph, hastad_thesis}, Fourier analytic methods \cite{lmn}, and polynomial methods \cite{razborov, smolensky}.

So how do we compare the quantum and classical settings? And what does this comparison tell us about the power of quantum circuits in constant depth? To start, classical circuits classes (e.g., $\NC^0$, $\AC^0$, $\TC^0$, \ldots) typically assume that the output of any gate can be used  as input for any number of other other gates (i.e., a gate's output can be ``fanned out'' to other gates). Of course, this is exactly the kind of fanout that immediately becomes so powerful when given to a constant-depth quantum circuit.

In fact, because of this fanout, the classical Threshold gate reigns supreme amongst similar classical circuit complexity classes. This is due to the fact that constant-depth classical circuits with Fanout and Threshold can compute any Boolean function where the output depends only on the Hamming weight of the input.\footnote{To see this, first notice that for any $k$, there is a constant-depth circuit with two Threshold gates that computes whether or not the input has Hamming weight exactly $k$. Since any symmetric Boolean function can be expressed as a disjunction over these ``exact-$k$'' clauses, the claim immediately follows due to the fact that a threshold of $1$ is equivalent to the Or function.} Formally, the complexity class $\TC^0$, which contains all languages computed by constant-depth classical circuits with Threshold, contains all other similarly defined classical circuit classes with other large fan-in gates: $\NC^0[p]$, $\AC^0$, $\AC^0[p]$, and $\ACC$.\footnote{See \Cref{sec:circuit_basics} for precise definitions.}  In many cases, Threshold is provably more powerful, e.g., $\AC^0 \subsetneq \TC^0$ \cite{ajtai_parity, fss_ac0} and $\AC^0[p] \subsetneq \TC^0$ \cite{razborov, smolensky}.

This is why the Threshold gate was a tantalizing target for quantum exploration. Prior to our work, it was \emph{not} known whether the quantum version of $\TC^0$---i.e., $\BQTC^0$---was as powerful as the quantum versions of the other classical complexity classes. In fact, given the surprising power of Fanout in the quantum world, the exact opposite was known: $\BQNC^0_{wf} \supseteq \BQTC^0$ \cite{hs_fanout}. That is, constant-depth quantum circuits with Fanout could simulate constant-depth circuits with Threshold. Our work restores order to the usually classical hierarchy, placing Threshold alongside Fanout as one of the most powerful quantum gates in constant depth: $\BQTC^0 = \BQNC^0_{wf}$.

\subsection{Proof techniques and overview}

The constructions in \Cref{thm:intro_main} and \Cref{thm:intro_general} follow a general outline pioneered by Rosenthal~\cite{rosenthal}. There, it is shown that constant-depth quantum circuits can compute Fanout using generalized Toffoli gates \emph{provided exponential-sized circuits are allowed.} While not phrased in this language, Rosenthal's construction shows a proof-of-principle technique for taking a very ``weak'' Parity gate (indeed, Toffoli non-trivially computes Parity for exactly one input!) and boosting it to a full Parity gate. We show that when we start with a gate (like Threshold) which is closer to Parity, this construction can be altered to yield circuits of polynomial size.

The proof goes in two steps. First, define a certain cat-like state called a ``nekomata'' \cite{rosenthal}:
\[
\frac{\ket{0^n}\otimes \ket{\psi_0} + \ket{1^n}\otimes \ket{\psi_1}}{\sqrt{2}}
\]
where $\ket{\psi_0}$ and $\ket{\psi_1}$ are arbitrary states. Following a similar idea to that of Green et al.\  \cite{qacc}, such states can be used to compute Parity in constant depth using the the relative phase between the $\ket{0^n}$ and $\ket{1^n}$ part of the state.
    
Second, show there is an explicit constant-depth construction for a nekomata state. Here, we show the key ingredient is the ability to create a ``noisy'' version of a usual cat state, where the all-zeroes and all-ones outcomes have noticeably larger amplitudes than those on the other outcomes. Threshold gates are significantly better at this task than the Toffoli gates in Rosenthal's original construction. Finally, these states can be combined together (using Toffoli or Threshold gates) to form a high-fidelity nekomata state, completing the construction.

\subsection{Related work\label{PRIOR}}

Our work shares some similarity to that of \cite{grilo:2024_qtc}, where the authors explore quantum advantage with constant-depth quantum circuits. They also make a similar claim suggesting that $\QTC^0 = \QNC^0_{wf}$, but crucially, their results hold in a circuit model with intermediate measurements and classical fanout. The classical fanout in their circuit model allows them to bootstrap the poor man's cat state construction of Bene Watts et al.\ \cite{watts2019exponential} to construct an actual cat state, an idea that was also explored in \cite{buhrman2023state}. To be clear, our circuit model and definition of $\BQTC^0$ follows in a traditional line of work (e.g, \cite{moore:1999, qacc, hs_fanout, tt_fanout, depth2, rosenthal, pauli_spec}), where no such intermediate measurements or classical fanout is allowed. Therefore, we must use entirely different techniques.

\subsection{Future directions}
One immediate open question left open by our work is whether the approximation error inherent in the construction used to prove \Cref{thm:intro_main} can be eliminated without incurring a size or depth blow-up.
More generally, we ask which other conditions on a family of multi-qubit gates lead to powerful shallow circuits. One explicit approach would be to ask what properties of the sets $S$ parameterizing our phase gates $U_S$ are sufficient to compute Fanout. Is there something beyond being parity restricted?

Another interesting question concerns the circuit complexity of restricted families of threshold functions. Specifically, consider the Exact-$k$ gate, which indicates if the Hamming weight of the input is exactly $k$. Notice that Exact-$k$ can be constructed from two Threshold gates. Moreover, for $k \approx n/2$, Exact-$k$ can be used to compute Threshold. This latter statement is not obvious and follows from the fact that our proof of \Cref{thm:intro_main} actually uses Exact gates rather than Threshold gates. However, for other values of $k \ll n/2$, it is not simple to see how Exact-$k$ could be used to simulate Exact-$(k+1)$.

\section{Preliminaries}
We will now introduce the different types of entangling gates considered in this work, the types of circuits constructed from them, and the complexity classes to which they roughly correspond. 
\subsection{Multiqubit Gates}
A simple multiqubit gate is the $\CNOT$ gate which acts on two qubits, flipping the target conditioned on the control, i.e.,
\begin{align*}
    \CNOT\ket{x_1, x_2} = \ket{x_1, x_1 \oplus x_2}
\end{align*}
Any two-qubit gate can be constructed from constantly many single-qubit gates and $\CNOT$ gates. A circuit consisting entirely of arbitrary single- and two-qubit gates is said to be a $\QNC$ circuit.

Another multi-qubit gate of interest is the Toffoli gate, which acts on three qubits by flipping the last qubit controlled on the first two, i.e.,
\begin{align*}
    \text{Tof}\ket{x, y, z} = \ket{x, y, (x \land y) \oplus z}
\end{align*}
This gate can be seen as a $\CNOT$ gate with an additional control qubit. In fact, we call the analogous unitary on $n > 1$ qubits a \emph{generalized} Toffoli gate:
\begin{definition}\label{definition:gen-toffoli}
    The generalized Toffoli gate $\land_n$ acts on $n + 1$ qubits by computing the $\AND$ of the first $n$ bits in superposition. For all $x_1, x_2, \dots x_n, b \in \{0, 1\}$ the $\wedge_n$-gate acts as
    \begin{align*}
    \land_n\ket{x_1, x_2, \dots x_n, b} = \ket{x_1, x_2, \dots x_n, (x_1 \land \cdots \land x_n) \oplus b}
\end{align*}
\end{definition}
Circuits composed of arbitrary single-qubit gates and generalized Toffoli gates are referred to as $\QAC$ circuits.
\begin{definition}\label{definition:threshold-gate}
    For $k \in \{0, 1 \dots n\}$ and $x_1, x_2, \dots x_n, b \in \{0, 1\}$ the unitary $\Threshold_{n, k}$ acts as
    \begin{align*}
        \Threshold_{n, k}\ket{b}\ket{x} = \ket{b \oplus \mathbb{I}_{|x| \geq k}}\ket{x}
    \end{align*}
\end{definition}
Circuits composed of arbitrary single-qubit gates and threshold gates\footnote{Recall that a function $f\colon \{0, 1\}^n \to \{0, 1\}$ is said to be a \emph{threshold function} if it  can be written as
\begin{align*}
    f(x) = \begin{cases} 
      1 & \sum_{i = 1}^n w_ix_i \geq b \\
      0 & \text{otherwise} \end{cases}
\end{align*}
for some $w_1, \dots w_n, b \in \mathbb{R}$. For our purposes it suffices to only consider threshold functions in which $w_i = 1$ for all $i \in [n]$.} are said to be $\QTC$ circuits. Note that by taking $k = n$ we recover the generalized Toffoli gate, and in this sense the generalized Toffoli gate \emph{is} a Threshold gate, so including this gate in the allowed gate-set for $\QTC$ circuits would be redundant.

Let $U_f$ be the unitary which computes some boolean function $f: \{0, 1\}^n \to \{0, 1\}$ in superposition i.e. for all $x \in \{0, 1\}^n$ and $b \in \{0, 1\}$ $U_f\ket{x, b} = \ket{x, b \oplus f(x)}$. Note that all multiqubit gates discussed thus far fall into this category. Now, observe that when the target qubit is replaced with $\ket{-} = \frac{\ket{0} - \ket{1}}{\sqrt{2}}$, we can ``compute $f$ in the phase":
\begin{align*}
    U_f\ket{-}\ket{x} = (-1)^{f(x)}\ket{-}\ket{x}
\end{align*}
So, given $U_f$ we can with a single single ancilla implement $V_f$ which acts as $V_f\ket{x} = (-1)^{f(x)}\ket{x}$. While going from $U_f$ to $V_f$ is not a difficult task the converse could in general be quite non-trivial.\footnote{For instance, take $Z^{\otimes n}$; this gate computes parity in the phase as $Z^{\otimes n}\ket{x} = (-1)^{|x|}\ket{x}$, but it is unclear if there is a simple way to recover the usual parity gate: $\Parity_n$.}

As mentioned, the \textit{quantum Fanout gate} gives us some way of ``copying" a given qubit and $\XOR$-ing it onto an unbounded number of qubits. 
\begin{definition}\label{definition:fanout} 
For all $x_1, x_2, \dots x_n, b \in \{0, 1\}$ the Fanout unitary, $\Fanout_n$, acts as
\begin{align*}
    \Fanout_n\ket{b}\ket{x_1, x_2, \dots x_n} = \ket{b \oplus x_1, b\oplus x_2, \dots b\oplus x_n}
\end{align*}
\end{definition}
We will refer to circuits constructed from one- and two-qubit and Fanout gates as $\QNC_{wf}$ circuits.

Another important class of gates are so-called $\MOD$ gates:
\begin{definition}
    For a given $m \in \mathbb{N}$ and all $x_1, x_2, \dots x_n, b \in \{0, 1\}$ the $\MOD_{n, m}$ gate acts as
    \begin{align*}
        \MOD_{n, m}\ket{x_1, x_2, \dots x_n}\ket{b} = \ket{x_1, x_2, \dots x_n}\ket{\Mod_{n, m}(x)\oplus b}
    \end{align*}
    Where $\Mod_{n, m}(x) = 1$ iff $|x|$ is divisible by $m$. Further, for $\ell \in \{0, 1, \dots m - 1\}$ we use $\Mod_{n, m, \ell}(x)$ to denote the function which is $1$ iff $|x|\equiv \ell\;(\mathrm{mod}\; m)$ and the corresponding quantum gate accordingly:
    \begin{align*}
        \MOD_{n, m, \ell}\ket{b}\ket{x_1, x_2, \dots x_n} = \ket{\Mod_{n, m, \ell}(x)\oplus b}\ket{x_1, x_2, \dots x_n}.
    \end{align*}
\end{definition}
Note that when $m = 2$ the $\MOD_{n, 2, 1}$ gate is equivalent to the parity gate $\Parity_n$. When a circuit consists of one- and two-qubit gates and $\MOD_{n, m}$ gates for a fixed $m$ it is called a $\QNC[m]$ circuit and when the circuit also contains generalized Toffoli gates it is referred to as a $\QAC[m]$ circuit.

The final class of gates we will define are what we call ``parity-restricted" gates which have not previously appeared in the literature. A set of bit strings $S \subseteq \{0, 1\}^n$ is said to be \emph{parity restricted} if $|s_1| \equiv |s_2| \pmod{2}$ for all $s_1, s_2 \in S$.
\begin{definition}\label{definition:parity_restricted_gates}
    A unitary $U_S$ acting on $n$-qubits is said to be a parity-restricted gate if for all $x\in\{0, 1\}^n$ 
    \begin{align*}
        U_S\ket{x} = (-1)^{\mathbb{I}_S(x)}\ket{x}
    \end{align*}
    for some parity-restricted set $S \subseteq \{0, 1\}^n$.
\end{definition}
A circuit composed of arbitrary one- and two-qubit gates and $U_S$ gates for some parity restricted set $S$ is said to be a $\QNC_S$ circuit. Similarly, if the circuit also consists of generalized Toffoli gates the circuit is said to be a $\QAC_S$ circuit. 

Finally, we will define the primary complexity measures for quantum circuits.
\begin{definition}\label{definition:depth}
    A quantum circuit $C$ is said to have \emph{depth} $d$ if $C$ can be decomposed as a sequence $M_dS_d \cdots M_2S_2M_1S_1$ where each $S_i$ consists entirely of single-qubit gates and $M_i$ consists of non-overlapping multi-qubit gates (i.e., every pair of gates in $M_i$ operate on disjoint sets of qubits).
\end{definition}
\begin{definition}
    A quantum circuit $C$ has \emph{size} $s$ if $C$ has exactly $s$ multi-qubit gates.
\end{definition}

\subsection{Quantum Circuit Complexity Classes}
In this section we will define the relevant quantum circuit classes, but before doing that we must introduce the notion of a circuit family and what it means for a circuit family to compute a Boolean function. 
\begin{definition}
    A family of quantum circuits is a collection $\mathcal{C} = \{C_n\}_{n \geq 1}$ where $C_n$ acts on $n + a(n)$ qubits where $a(n)$ is some computable function.
\end{definition}
This definition of a circuit family is analogous to the classical notion of a $\emph{non-uniform}$ circuit family since there need not be any relation between circuits for different sizes (e.g., it is not necessary for there to exist a Turing machine which outputs a description of $C_n$ on input $1^n$. Such a requirement is only for \emph{uniform} circuit families). It should be noted that all constructions presented in this work correspond to uniform circuit families nonetheless.

\begin{definition}\label{definition:circuit-comp}
For a given language $L \subseteq \{0, 1\}^*$ we say that a family of quantum circuits $\{C_n\}_{n \geq 1}$ each acting on $n + a(n)$ qubits exactly computes $L$ if for all $n \geq 1$ and $x \in \{0, 1\}^n$ measuring the last qubit of $C_n\ket{x}\ket{0^{a(n)}}$ in the computational basis yields 
\begin{itemize}
    \item $\ket{1}$ with certainty if $x \in L$
    \item $\ket{0}$ with certainty if $x \not\in L$
\end{itemize}
\end{definition}
Now, for the complexity classes of interest:
\begin{itemize}
    \item $\QNC^i$ is the class of problems decidable by $\QNC$ circuits which act on polynomially-many qubits (i.e. $n + a(n)$ is bounded by some polynomial in $n$), have polynomial size and depth $O(\log^i(n))$.
    \item $\QAC^i$ is the class of problems decidable by $\QAC$ circuits which act on polynomially-many qubits, have polynomial size and depth $O(\log^i(n))$.
    \item $\QTC^i$ is the class of problems decidable by $\QTC$ circuits which act on polynomially-many qubits, have polynomial size and depth $O(\log^i(n))$.
    \item $\QNC^i_{wf}$ is the class of problems decidable by $\QNC_{wf}$ circuits which act on polynomially-many qubits, have polynomial size and depth $O(\log^i(n))$.
\end{itemize}

The primary focus of this work will be constant depth circuits, which correspond to $i = 0$ in the above definitions i.e. the classes $\QNC^0$, $\QAC^0$, $\QTC^0$, and $\QNC_{wf}^0$. In a slight abuse of notation we may call a family of circuits a $\mathcal{C}$-circuit family if the family satisfies the necessary conditions for circuits which compute languages in $\mathcal{C}$ for some circuit class $\mathcal{C}$, though the unitaries which these circuits compute may not actually correspond to a Boolean function. For instance if $\{C_n\}_{n \geq 1}$ is a family of constant-depth, polynomial-size $\QAC$ circuits which act on polynomially many qubits we may refer to them simply as a family of $\QAC^0$ circuits. 

\begin{proposition}[Proposition 3.1 of \cite{qacc}]\label{EQUIV}
    The following tasks are equivalent for constant-depth circuits consisting of $\wedge_n$-gates and single-qubit gates:
    \begin{enumerate}
        \item Preparing the state $\frac{\ket{0^n} + \ket{1^n}}{\sqrt{2}}$ from $\ket{0^n}$ and performing the inverse transformation.
        \item Applying Fanout $\Fanout_n$.
        \item Applying Parity $\Parity_n$.
    \end{enumerate}
    In other words, these tasks are equivalent under $\QAC^0$ reductions.
\end{proposition}

Critical to our construction is the fact that (1) in the above proposition can be relaxed to a more general state preparation task. To see how, we must define a class of a ``cat-like" states, first introduced by Rosenthal \cite{rosenthal} which he calls \emph{nekomata}:
\begin{definition}\label{NEKOMATA}
    A state $\ket{\phi}$ on $n + m$ qubits is said to be an $n$-\emph{nekomata} if there exists some ordering of the qubits such that 
    \begin{align*}
        \ket{\phi} = \frac{\ket{0^n}\otimes \ket{\psi_0} + \ket{1^n}\otimes \ket{\psi_1}}{\sqrt{2}}
    \end{align*}
    where $\ket{\psi_0}$ and $\ket{\psi_1}$ are arbitrary $m$-qubit states. The first $n$ qubits of this state are referred to as the target qubits.
\end{definition}
As mentioned, \Cref{EQUIV} is still true when the cat state in task 1 is replaced with any $n$-nekomata (see \Cref{PROOFS} for more details). This fact is quite powerful since we only need to design a circuit which produces a state on which some subsystem is ``cat-like" in order to compute parity. This makes the prospect of designing a circuit to compute parity far less daunting.  
\label{sec:circuit_basics}

\subsection{Approximate Quantum Circuits}\label{APPROX}
Proposition \ref{EQUIV} shows that exactly preparing a cat state is in fact computationally equivalent to exactly computing parity, up to some $\QAC^0$ computations and this can further be generalized by relaxing the task of preparing a nekomata state. 
Further, it is established in \cite{rosenthal} that preparing an approximate nekomata state is sufficient to \emph{approximately} compute parity or fanout. This notion is made precise below. 
\begin{definition}\label{ANEKOMATA}
    For $\epsilon \in [0, 1]$ a state $\ket{\phi}$ on $n + m$ qubits is said to be an $\epsilon$-approximate nekomata if there exists some nekomata $\ket{\nu}$ such that $\left|\braket{\nu}{\phi}\right|^2 \geq 1 - \epsilon$.
\end{definition}
When we refer to a quantum circuit as approximately computing some function or approximating a given unitary we mean that the circuit, $C$, and the ideal unitary $U$ have small distance. Explicitly, for $\epsilon \in (0, 1)$ we say that $C$ is an $\epsilon$-approximate implementation of $U$ or that $C$ computes $U$ with approximation error $\epsilon$ if $\| U -  C\|_{\mathrm{op}} \leq \epsilon$ where $\|\cdot\|_{\mathrm{op}}$ denotes the operator norm. 

A statement analogous to \Cref{EQUIV} holds for the approximate version of each task:
\begin{lemma}[Theorem 3.1 of \cite{rosenthal}] \label{approx-NEK-PAR-FAN}
    For any $\epsilon \in (0, 1)$ the following tasks are equivalent under $\QAC^0$ reductions:
    \begin{itemize}
        \item Preparation of $O(\epsilon)$-approximate nekomata from the all zeros state and the inverse transformation
        \item Approximately computing Parity with error $O(\epsilon)$
        \item Approximately computing Fanout with error $O(\epsilon)$
    \end{itemize}
\end{lemma}
This lemma is again quite useful for us as circuit designers; now any circuit producing a state which has some subsystem that is \textit{approximately} cat-like suffices to approximately implement fanout or parity.

Finally, we define the bounded-error analogues of the quantum circuit complexity classes introduced thus far:
\begin{definition}[$\BQNC^i$]
A decision problem $L \subseteq \{0, 1\}^*$ is in $\BQNC^i$ if there exists a family of $\QNC^i$ circuits $\{C_n\}_{n \in \mathbb{N}}$ acting on $n + a(n) = \poly(n)$ qubits and a constant $c > 0$ such that for all $n\in \mathbb{N}$ and $x \in \{0, 1\}^n$ measuring the last qubit of $C_n\ket{x}\ket{0^{a(n)}}$ in the computational basis yields
\begin{itemize}
    \item $\ket{1}$ with probability at least $2/3$ if $x \in L$
    \item $\ket{1}$ with probability at most $1/3$ if $x \not\in L$
    \end{itemize}
\end{definition}

$\BQAC^i$, $\BQTC^i$, and $\BQNC^0_{wf}$ are defined similarly for their respective circuit classes.

\section{Bootstrapping weak parity gates\label{PMAJ}}
In this section we will show that for any non-empty parity restricted set $S \subseteq \{0, 1\}^n$ the unitary $U_S\ket{x} = (-1)^{\mathbb I_{x \in S}}\ket{x}$ can be bootstrapped in constant depth to approximately compute Parity. This construction generalizes the constant-depth exponential-size $\QAC$ circuit family given in \cite{rosenthal}. As a corollary, we find that for any polynomial $p$, there exist poly-size $\QTC^0$ circuits which have fidelity $1 - 1/p(n)$ with Parity. 

\subsection{Grid Construction}
Rather than directly computing Parity, the circuits described in this section will prepare approximate nekomata, which via \Cref{approx-NEK-PAR-FAN} can be used to compute Parity and Fanout with high essentially the same approximation error, up to constant factors.

We will make use of the following lemma:
\begin{lemma}[Lemma 4.3 of \cite{rosenthal}]\label{approx-nek}
Let $\ket\varphi$ be a state with $n$ ``target" qubits that measure to all-zeros with probability at least $1/2 - \epsilon$ and all-ones with probability at least $1/2 -\epsilon$. Then there exists an $n$-nekomata $\ket\nu$ such that $|\braket{\nu}{\varphi}|^2 \ge 1-2\epsilon$.
\end{lemma}
\begin{proof}
    Suppose that the first $n$ qubits of $\ket\varphi$ are the targets. Then, the state 
    \[
    \ket\nu = \frac{1}{\sqrt{2}}\sum_{b \in \{0, 1\}}\frac{\proj{b^n}\otimes \mathbb{I}\ket{\varphi}}{\|\proj{b^n}\otimes \mathbb{I}\ket\varphi\|}
    \]
    is an $n$-nekomata and
    \begin{align*}
        |\braket{\varphi}{\nu}|^2 = \bigg{(}\frac{1}{\sqrt{2}}\sum_{b \in \{0, 1\}}\|\proj{b^n}\otimes \mathbb{I}\ket\varphi\|\bigg{)}^2 \geq \frac{1}{2}\bigg{(}\sqrt{1/2 - \epsilon} +  \sqrt{1/2 - \epsilon} \bigg{)}^2 = 1 - 2\epsilon
    \end{align*}
\end{proof}

As mentioned, a parity restricted gate can be though of a ``weak" parity gate in the sense that it correctly computes parity on some fraction $\frac{1 + \epsilon}2$-fraction of the inputs. The idea behind our construction is to use these ``weak" parity gates to prepare many bad, but not horrible approximate cat states in parallel. These bad, but not horrible cat states are of the form
\begin{align*}
    \ket{\phi} = \sqrt{p_0}\ket{0^n} + \sqrt{p_1}\ket{1^n} + \sqrt{\epsilon}\ket{\omega}
\end{align*}
where $\ket\omega$ is orthogonal to $\ket{b^n}$ for $b \in \{0, 1\}$. These initial states are bad approximate cat states in the sense that they have little overlap with the cat state, but aren't horrible because the distributions corresponding to their measurement outcomes are peaked only at $\ket{0^n}$ and $\ket{1^n}$. In the final stage of our construction we accrue the distributions on each of these bad cat states into some $n$ target qubits using Toffoli gates. We then show that for the right choice of parameters this accruing step effectively amplifies the original bad, but not horrible, distribution given by each of the weak parity gates in parallel. The final result is a good approximate nekomata.
\begin{figure}
    \centering
    \input{figures/nekomata_construction}
    \caption{Constructing a nekomata from $U_S$ and Toffoli gates. Target qubits are shown in blue.}
    \label{fig:nekomata_construction}
\end{figure}
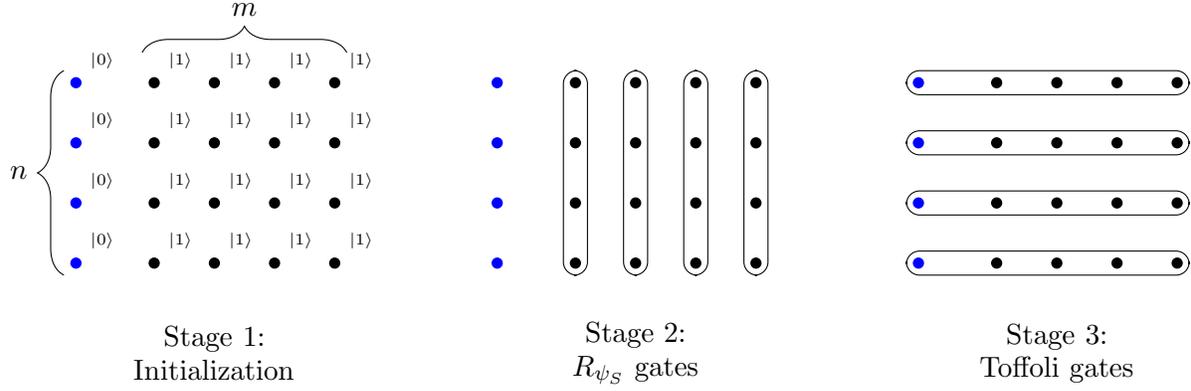
\begin{theorem}\label{GRID}
    For any parity-restricted $S \subseteq \{0, 1\}^n$ with $|S| \leq 2^{n - 4}$ there exists a depth-$4$, $O(n + \frac{2^{2n}}{|S|^2})$-size $\QAC_S$ circuit that constructs an $O(|S|^22^{-2n})$-approximate nekomata.
\end{theorem}
\begin{proof}
    The $\QAC_S$ circuit will act on $n(m + 1)$ qubits arranged in a grid of width $m + 1$ and height $n$. The first column will be designated as the ``target" qubits, initialized to $\ket{0^n}$ and all other columns initialized to $\ket{1^n}$ (say, with a layer of $X$ gates). To each column apply an $R_{\psi_S} = \mathbb{I} - 2\proj{\psi_S}$ gate, where $\ket{\psi_S} = H^{\otimes n}U_SH^{\otimes n}\ket{0^n}$. Note that this can be implemented in depth-$3$ as 
    \[
    \mathbb{I} - 2\proj{\psi_S} = H^{\otimes n}U_SH^{\otimes n}(\mathbb{I} - 2\proj{0^n})H^{\otimes n}U_SH^{\otimes n}
    \]
    which looks like the following quantum circuit:
    
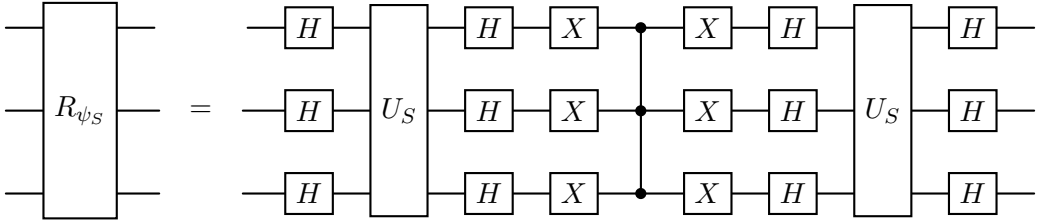
\begin{center}
    \begin{quantikz}
     & \gate[3]{R_{\psi_S}} & \midstick[3,brackets=none]{\;\;=\;\;} & \gate{H} & \gate[3]{U_S} & \gate{H} & \gate{X} & \ctrl{1} & \gate{X} & \gate{H} & \gate[3]{U_S} & \gate{H} &  \\
     &  &  & \gate{H} &  & \gate{H} & \gate{X} & \ctrl{1} & \gate{X} & \gate{H} &  & \gate{H} & \\
     &  &  & \gate{H} &  & \gate{H} & \gate{X} & \control{} & \gate{X} & \gate{H} & & \gate{H} & 
    \end{quantikz}
    \end{center}
    Finally, apply a Toffoli gate along each row with the output qubit being the corresponding target qubit (i.e. the qubit in the first column). We will now show that the probability that the target column is measured (in the computational basis) as $\ket{b^n}$ is at least $\frac{1}{2} - \epsilon$ for $b \in \{0, 1\}$.
    
    To start, let
    \begin{align*}
    \gamma_0 &:= \braket{0^n}{\psi_k}^2 
        = \bra{0^n}H^{\otimes n} U_S H^{\otimes n}\ket{0^n}^2 
        = \bigg{(} 1 - \frac{|S|}{2^{n - 1}}\bigg{)}^2 \\
    \gamma_1 &:= \braket{1^n}{\psi_k}^2
        = \bra{1^n}H^{\otimes n} U_S H^{\otimes n}\ket{0^n}^2
        = \frac{|S|^2}{2^{2n - 2}}
    \end{align*}
   For $b \in \{0, 1\}$, let $p_b$ be the probability that a given non-target column yields $\ket{b^n}$ after measuring in the computational basis. We have
   \begin{align*}
        p_0 &= \bra{0^n} (\mathbb{I} - 2 \proj{\psi_k})\ket{1^n}^2 = 4\gamma_0\gamma_1 \\\
        p_1 &= \bra{1^n} (\mathbb{I} - 2 \proj{\psi_k})\ket{1^n}^2 = (1 - 2\gamma_1)^2 
    \end{align*}
    Note that computational basis measurements commute with Toffoli gates of any size, so in order for the targets to be measured as $\ket{1^n}$ all other columns must also be measured as $\ket{1^n}$. Let $m = \lfloor \frac{-\ln(2)}{2\ln(1 - 2\gamma_1)} \rceil$, then
    \begin{align*}
        \mathbb{P}[\text{Targets measure } \ket{1^n}] &= (1 - 2\gamma_1)^{2m}\\
        &> (1 - 2\gamma_1)^{\frac{-\ln(2)}{\ln(1 - 2\gamma_1)} + 1}\\
        &> \frac{1}{2}(1 - 2\gamma_1)\\
        &= \frac{1}{2} - \frac{|S|^2}{2^{2n - 2}}
    \end{align*}
    Now, call a non-target column ``bad" if it is measured as anything other than $\ket{0^n}$ or $\ket{1^n}$. Via a union bound,
    \begin{align*}
    \mathbb{P}[\text{Some column is bad}] \le m(1 - p_0 - p_1) = m(1 - 4\gamma_0\gamma_1 - (1 - 2\gamma_1)^2)
        = 4m\gamma_1(1 - \gamma_0 - \gamma_1)
    \end{align*}
    Observe that 
    \[
    \frac{1}{2}(1 - 2\gamma_1) = (1 - 2\gamma_1)^{\frac{-\ln(2)}{\ln(1 - 2\gamma_1)} + 1} \leq (1 - 2\gamma_1)^{2m} \leq \exp(-4m\gamma_1),
    \]
    so $4m\gamma_1\leq -\ln(\frac{1}{2} - \gamma_1) < 1$ as $\gamma_1 \leq \frac{1}{16}$. So,
    \begin{align*}
        4m\gamma_1(1 - \gamma_0 - \gamma_1) &< 1 -\gamma_0 -\gamma_1 \\
        &\leq \frac{|S|}{2^{n - 2}}
    \end{align*}
    Thus, every column is good with probability at least $1 - \frac{|S|}{2^{n - 2}}$ and the targets are measured as $\ket{0^n}$ with probability at least $1 - \frac{|S|}{2^{n - 2}} - (\frac{1}{2} - \frac{|S|^2}{2^{2n - 2}}) \geq \frac{1}{2} - \frac{|S|}{2^{n - 2}}$. By Lemma \ref{approx-nek}, the state produced is an $\frac{|S|}{2^{n -3}}$-approximate nekomata. Also, note that $m = \Theta(1/\gamma_1)$, hence the circuit has size $O(n + m) = O(n + \frac{2^{2n}}{|S|^2})$.
\end{proof}
Next, we argue that Threshold gates can be used to hit the ``sweet spot'' of \Cref{GRID}. Namely, they correspond to parity-restricted sets of the right size to make the size and accuracy of the above construction polynomial.

\begin{corollary}
    $\BQTC^0 = \BQNC^0_{wf}$
\end{corollary}
\begin{proof}
    Without a loss of generality assume $n$ is even, otherwise this can be rectified with a single ancilla qubit. Let $S = \{x \in \{0, 1\}^n \ : \ |x| = \frac{n}{2}\}$ i.e. $S$ is the Hamming slice of weight $\frac{n}{2}$. Note that $|x| = \frac{n}{2}$ if and only if $\Maj(x_1, \dots, x_n) = 1$ and $\Maj(x_1, \dots x_n, 0) = 0$, thus the $U_S$ gate can be implemented in depth $2$ by applying majority once to $\ket{\psi}$, tacking on an ancilla qubit set to $\ket{0}$ then applying another majority gate to $\ket{\psi}\ket{0}$. In this case the circuit from Theorem \ref{GRID} has size $s = O(n + \frac{2^{2n}}{|S|^2})$ and produces an $\epsilon = \frac{|S|}{2^{n - 3}}$-approximate nekomata. Since 
    \begin{align*}
        \dbinom{n}{n/2} \in \bigg{[}\frac{2^n}{\sqrt{2n}},  \frac{2^n}{\sqrt{n\pi/2}} \bigg{]}
    \end{align*}
    via Stirling's formula, it follows that $s = O(n)$ and $\epsilon = O(\frac{1}{\sqrt{n}})$. 
    
    Note that we can make the error an arbitrarily small polynomial by making the circuit polynomially larger. That is, suppose we want to construct an $O(n^{-c/2})$-approximate $n$-nekomata for some $c > 1$. Simply use the construction above for an $n^c$-nekomata, but only use $n$ of the targets (i.e., any $n^c$-nekomata \emph{is} an $n$-nekomata) The circuit will have size $O(n^c)$ and error at most $O(n^{-c/2})$. Therefore, by \Cref{approx-NEK-PAR-FAN}, there are poly-size $\QTC^0$ circuits to compute Fanout to arbitrary polynomial precision, so $\BQTC^0 = \BQNC^0_{wf}$.
\end{proof}
\subsection{Removing the Toffoli gates}
The construction presented in Theorem \ref{GRID} for generic $S$ requires large Toffoli gates, however we will show that for some regimes of $|S|$, these gates are unnecessary, i.e., $\QNC^0_S$ circuits can exactly compute Toffoli on polynomially many qubits. We will show that this is indeed the case when $|S| \geq 2^{n - O(1)}$ and $|S| \leq 2^{(1 - \epsilon)n}$ for a fixed constant $\epsilon < 1$.
\begin{lemma}\label{SUBS}
    Let $c$ be constant, and let $S$ be a parity restricted set with size $|S| \geq 2^{n - c}$ and strings of parity $b \in \{0, 1\}$. Then, there exist some $c - 1$ bit-strings $t_1, t_2, \dots t_{c - 1} \in \{0, 1\}^n$ such that $|x| \equiv b \pmod{2}$ iff 
    \begin{align*}
        \bigvee_{y \in \mathrm{Span}(t_1, \dots t_c)} \{x \oplus y \in S\}
    \end{align*}
    is satisfied.
\end{lemma}
\begin{proof}
    Suppose that $b = 0$. Let $\mathcal E \subset \mathbb{F}_2^n$ be the subspace of dimension $n - 1$ consisting of vectors of even Hamming weight. Observe that $S \subseteq \mathcal E$, and moreover that $S$ contains at least $n - c$ linearly independent elements of $\mathcal E$. Therefore, there exist some $t_1, \cdots t_{c - 1} \in \mathcal E$ such that $S \cup \{t_1, \dots t_{c - 1}\}$ span $\mathcal E$. Hence, every element of $\mathcal E$ can be written as $s \oplus y$ for some $s \in S$ and $y \in \text{span}\{t_1, \dots t_{c - 1}\}$. Thus, $|x| \equiv 0$ iff the disjunction is satisfied.
    
    If $b = 1$ then the set $S'$ obtained by flipping the first bit of every element of $S$ contains vectors of even Hamming weight - further, $S'$ contains at least $n - c$ linearly independent vectors. So, take $\{t_1, \dots t_{c - 1}\}$ as before such that $S' \cup \{t_1, \dots t_{c - 1}\}$ spans $\mathcal E$. Any vector of even Hamming weight, $y\in \mathbb{F}_2^n$, can be expressed as $y = s' + t'$ for some $s' \in S'$ and $t' \in \text{Span}(t_1, \dots t_c)$. Now, observe that if $x = (x_1, \dots x_n) \in \mathbb{F}_2^n$ has odd Hamming weight then $(x_1 \oplus 1, \dots x_n)$ can be expressed as $s' + t'$ for some $s' \in S'$ and $t'\in \text{Span}(t_1, \dots t_c)$. Since $s' = (s_1\oplus 1, s_2, \dots s_n)$ for some $s = (s_1, s_2, \dots s_n) \in S$, it follows that $x$ has odd Hamming weight iff the disjunction is satisfied.
\end{proof}
    As shown above, if $|S| \geq 2^{n - c}$ for some constant $c$ then we can extend $S$ linearly to ``cover" all strings of a fixed Hamming weight. To implement a parity gate in this way, one can encode every linear combination $t' \in \text{Span}(t_1, \dots t_c)$ in the ancilla and then apply a $U_S$ gate to $\ket{x \oplus t'}$ for every $t'$ - this can be done in constant depth using only $U_S$ and $\CNOT$ gates since  $|\text{Span}(t_1, \dots t_c)| = 2^{c} = O(1)$. If any of these $U_S$ gates evaluate to $1$ then $x$ must have Hamming weight consistent with that of the strings in $S$, in effect computing the parity of $x$. We will now see how generalized Toffoli can be computed with just $U_S$ and $\CNOT$ gates when $|S|$ is sufficiently small.
\begin{lemma}\label{RESTRICT}
    For any $S \subseteq \{0, 1\}^n$ there exists some $s \in S$ and some subset of indices $\{i_1, i_2, \dots i_k\}$ such that $s$ is the unique $x \in S$ which satisfies
    $x_{i_j} = s_{i_j}$ for all $j \in [k]$. Further, $k \leq \log{|S|}$.
\end{lemma}
\begin{proof}
    Note that unless $|S| = 1$ there exists some index on which elements of $S$ take different values. If $|S| = 1$ we are done and can take this single element to be $s$. Otherwise, let $i_1$ be the first index on which elements of $S$ take different values. We will now partition $S$ into two sets $S^1_0$ and $S^1_1$ where $S^1_b = \{s \in S \ | \ s_{i_1} = b\}$. Take $T_1$ to be the set with fewer elements and repeat this procedure, defining $T_j$ similarly for $j > 1$. Since $|T_{j + 1}| \leq \frac{|T_{j + 1}|}{2}$ for $j > 1$, it follows that for some $k > 1$, $|T_k| = 1$ and it follows that $k \leq \log{|S|}$.
\end{proof}
Now, when $|S|$ is sufficiently small, there is always some way to fix a small number ($\log |S|$) of bits so that the unfixed bits have a unique assignment consistent with $S$. In particular, if $|S| \leq 2^{(1 - \epsilon)n}$ for some constant $\epsilon \in (0, 1)$ then there exists some partial assignment of at most $(1 - \epsilon)n$ bits such that for the remaining $\epsilon n$ bits there is a unique assignment such that the resulting string is a member of $S$. For simplicity, suppose that the partial assignment is on the first $(1 - \epsilon)n$ as $y \in \{0, 1\}^{(1 - \epsilon)n}$ and that the unique assignment for the remaining bits which is consistent with $S$ is $z \in \{0, 1\}^{\epsilon n}$. Now, for $x \in \{0, 1\}^{\epsilon n}$ we can see that
\begin{align*}
    U_S\ket{y}\ket{x}\ket{0} = \ket{y}\ket{x}\ket{\mathbb{I}_{z}(x)}
\end{align*}
In this way, after fixing the first $(1 - \epsilon)n$ bits to $y$ forces $U_S$ to act like a $U_{\{z\}}$ gate on the remaining qubits. This gate is locally equivalent to a Toffoli gate; we can just apply $X$ gates to the wires on which $z_i = 0$. Since $\epsilon$ is a constant, we can repeat this procedure $1/\epsilon$ times to implement the Toffoli gate on $n$ qubits in constant depth. In this way, we can directly implement generalized Toffoli gates in the grid construction of \Cref{GRID} using the $U_{S_n}$ gates when $|S|$ is sufficiently small.

\begin{corollary}
    For any parity restricted set $S$ which satisfies $|S| \geq \Omega(2^n)$ or $|S| \leq 2^{(1 - \epsilon)n}$ for some fixed $\epsilon \in (0, 1)$ there exist constant-depth $\QNC_S$ circuits of size $O(n + \frac{2^{2n}}{|S|^2})$ which prepare $O(\frac{|S|^2}{2^{2n}})$-approximate nekomata. 
\end{corollary}

\section{Quantum \texorpdfstring{$\MOD$}{MOD} gates are powerful even on their own}
In this section we show a strengthening of a result of \cite{qacc}. In particular they show that for any fixed $q > 1$ $\MOD_q$ gates, $\wedge_n$ gates, single- and two-qubit gates can be leveraged to implement Fanout in constant depth and polynomial size: 
\begin{theorem}[Theorem 4.6 of \cite{qacc}]
    For $p \geq 2$, $\QAC^0[p] = \QAC^0_{wf}$.
\end{theorem}
However, the family of $\QAC^0[p]$ circuits they construct does not actually require $\wedge_n$ gates i.e. the family of circuits they construct to compute $\Fanout_n$ is actually a $\QNC^0[p]$ family. This immediately yields the collapse of all $\QNC^0[p]$:
\begin{theorem}{\label{NCMOD}}
    $\QNC^0[p] = \QNC^0_{wf}$ for all primes $p$.
\end{theorem}
Combined with the results of \cite{hs_fanout} and \cite{tt_fanout} we have an even larger collapse of constant-depth circuit classes:
$$
\QNC^0[p] = \QNC^0_{wf} = \QAC^0_{wf} = \QAC^0[q] = \QTC^0_{wf}
$$
for all $p, q \geq 2$. Before showing and analyzing the construction we will introduce some preliminaries.
\subsection{Simulating qudit arithmetic in \texorpdfstring{$\QNC^0[p]$}{QNC\^[p]}}
By a proposition of Moore, in order to inplement $\Fanout_n$ it actually suffices to construct a circuit which behaves like $\Fanout_n$ when all but one qubit are set to $\ket{0}$:
\begin{proposition}[Proposition 1 of \cite{moore:1999}]{\label{cat-to-fanout}}
    In any class of quantum circuits which includes Hadamard and $\CNOT$-gates, the follow are equivalent in constant depth:
    \begin{enumerate}
        \item \label{step:moore_tricky} It is possible to map $(\alpha\ket{0} + \beta\ket{1})\ket{0^{n - 1}}$ to $\alpha\ket{0^n} + \beta\ket{1^n}$ and from $\alpha\ket{0^n} + \beta\ket{1^n}$ to $(\alpha\ket{0} + \beta\ket{1})\ket{0^{n - 1}}$ for all $|\alpha|^2 + |\beta|^2 = 1$ 
        \item $\Fanout_n$ can be implemented with at most $n - 1$ ancilla qubits
        \item $\P_n$ can be implemented with at most $n - 1$ ancilla qubits
    \end{enumerate}
\end{proposition}
Hence, constructing a unitary $U$ which satisfies $U(\alpha\ket{0} + \beta\ket{1})\otimes\ket{0^{n - 1 + a(n)}} = (\alpha\ket{0^n} + \beta\ket{1^n})\ket{0^{a(n)}}$ for any single-qubit state $\alpha\ket{0} + \beta\ket{1}$ will result in the ability to compute Fanout; this is exactly what the construction does.

First, for a fixed prime $p$ consider the qudit generalizations of the Parity and Fanout gates for local dimension $p$:
\begin{align*}
    \M_{n, p}\ket{b}\ket{x_1x_2\cdots x_n} &= \ket{b - |x| \mod{p}}\ket{x_1x_2\cdots x_n}\\
    \Fanout_{n, p}\ket{b}\ket{x_1x_2\cdots x_n} &= \ket{b}\ket{(x_1 + b \mod{p}), (x_2 + b \mod{p}), \ldots, (x_n + b \mod{p})}
\end{align*}
where $x_1, \ldots, x_n, b \in \{0, \dots p - 1\}$.

Additionally, consider the following single-qudit gate:
\begin{align*}
    Q_p\ket{b} = \frac{1}{\sqrt{p}}\sum_{j = 0} \omega^{jb}\ket{j}
\end{align*}
where $\omega = e^{2i\pi/p}$. For example, when $p = 2$, $Q_p = H$, and $\Fanout_{n, p}$ and $\M_{n, p}$ are the usual Fanout and Parity (up to an $X$ gate on the output qubit) gates for qubits, respectively.
\begin{lemma}[Proposition 4.2 of \cite{qacc}]{\label{MODA}}
    $\M_{n, p} = (Q_p^{\dag})^{\otimes (n + 1)}\Fanout_{n, p}Q_p^{\otimes (n + 1)}$
\end{lemma}

Recall our goal: we want to use Mod-$p$ gate to simulate Fanout over qubits. While this seems somewhat challenging for qubits, it is trivial over qudits of local dimension $p$ by \Cref{MODA}. Therefore, our plan will be to pretend that we are in that setting by encoding a qudit using several qubits. Once we have set an encoding, we need encoded versions of the $\M_{n, p}$ and $Q_p$ gates in \Cref{MODA}. Encoding the $Q_p$ gate is easy---it's a gate of constant-size and each one of our encoded qudits will be of constant size, so any brute force encoding of $Q_p$ will do. The challenging step is to show that an encoded $\M_{n, p}$ gate is possible using (qubit) $\MOD_{n, p}$ gates. One of the key observations is that after we've applied the encoded Fanout, we will have accomplished Task~\ref{step:moore_tricky} of \Cref{cat-to-fanout}, and therefore, we can construct general Fanout over qubits.

\begin{proof}[Proof of Theorem \ref{NCMOD}]
    To start, define a linear encoding map $E \colon \mathbb C^p \to (\mathbb C^2)^{\otimes p}$ which maps from qudits of local dimension $p$ to a tensor product of $p$ qubits: 
    \begin{align*}
        E \ket{j} = \bigotimes_{k = 0}^{p - 1}\ket{\delta_{k, j}}
    \end{align*}
    for all $j \in \{0,\ldots, p-1\}$ and where $\delta_{k, j}$ denotes the Kronecker delta function. Note that $E$ is a linear map of full rank from the $p$-dimensional space spanned by $\{\ket j\}_{j = 0}^{p - 1}$ to the $p$-dimensional subspace of $(\mathbb{C}^2)^{\otimes n}$ spanned by$\bigg{\{}\bigotimes_{k = 0}^{p - 1}\ket{\delta_{k, j}}\bigg{\}}_{j = 0}^{p - 1}$. Throughout this construction we will only be working over qubits and not actually implementing $E$. Instead, we will be exploiting the equivalence of \Cref{MODA} by simulating qudit arithmetic with qubits. We introduce $E$ for the sake of describing this simulation method succinctly.
    
    Now, applying $Q_p$ to an encoded state on qubits amounts to implementing any $\Tilde{Q}_p$ which satisfies:
    $$
    \Tilde{Q}_p\ket{\delta_{0, j}}\ket{\delta_{1, j}}\cdots \ket{\delta_{p - 1, j}} = \frac{1}{\sqrt{p}}\sum_{k = 0}^{p - 1}\omega^{jk}\ket{\delta_{0, k}}\ket{\delta_{1, k}}\cdots \ket{\delta_{p - 1, k}}
    $$
    i.e., any unitary $\Tilde{Q}_p$ on $(\mathbb{C}^2)^{\otimes p}$ which respects the homomorphism induced by $E$. Since $p$ is fixed, any such $\Tilde{Q}_p$ operates on constantly many qubits and can be implemented in constant depth and size.
 
    Now, we must implement $\M_{n, p}$ on the encoded subspace using just $\MOD_{n, p}$ (and one- and two-qubit gates). Note that for any $j_1, \dots j_n \in \mathbb{F}_p$ their sum modulo $p$ can be decomposed
    \begin{align*}
        j_1 + \cdots + j_n &\equiv 0(\delta_{0, j_1} + \cdots + \delta_{0, j_n}) + 1(\delta_{1, j_1} + \cdots + \delta_{1, j_n}) + \cdots + (p - 1)(\delta_{p - 1, j_1} + \cdots + \delta_{p - 1, j_n}) \\
        &\equiv \sum_{k = 0}^{p - 1}k\bigg{(}\sum_{i = 1}^n \delta_{k, j_i} \bigg{)} \pmod{p}
    \end{align*}
    Let $s_k := \sum_{i = 1}^n \delta_{k, j_i} \pmod{p}$ be the number of $j_i$ terms equal to $k$ modulo $p$. Let's also define a family of generalized Mod-$p$ gates over qubits.
    For $\ell \in \{0, \dots p - 1\}$, recall $\MOD_{n, p, \ell}$ acts as
    \begin{align*}
        \MOD_{n, p, \ell}\ket{b}\ket{x_1, \dots x_n} = \ket{b \oplus \Mod_{n, p, \ell}(x)}\ket{x_1, \dots x_n}
    \end{align*}
    Notice that $\MOD_{n, p, \ell}$ can be implemented over qubits with an $\MOD_{n, p, 0}$ gate (the standard $\MOD_{n, p}$ gate on qubits) and $p - 1$ additional ancilla qubits, $p - \ell$ of which are set to $1$. We can use these gates to compute $s_k$ for a given $k \in \{0, 1 \dots p - 1\}$:
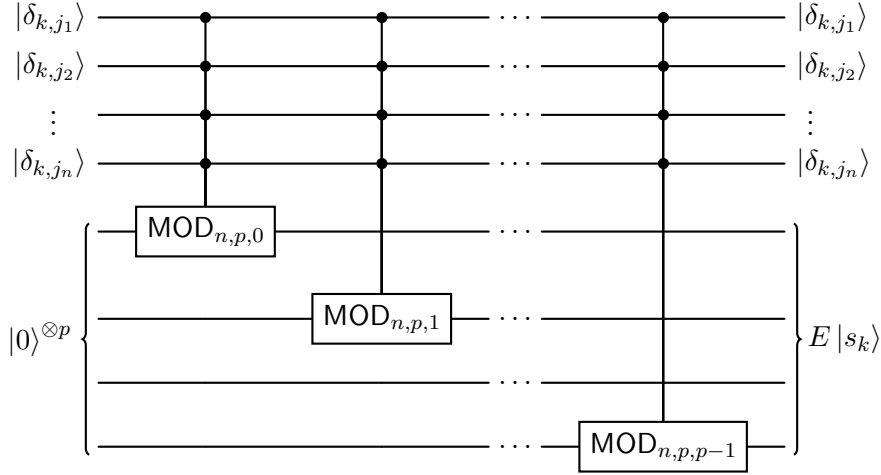
\begin{center}
    \begin{quantikz}
    \lstick{\ket{\delta_{k, j_1}}}  &  \ctrl{4} & \ctrl{5} & \ \ldots\  & \ctrl{7} & \rstick{\ket{\delta_{k, j_1}}}\\
    \lstick{\ket{\delta_{k, j_2}}} &\ctrl{3} & \ctrl{4} & \ \ldots\  & \ctrl{6} & \rstick{\ket{\delta_{k, j_2}}}\\
    \lstick{\vdots \ \ \ } & \ctrl{2} & \ctrl{3} & \ \ldots \  &\ctrl{5} & \rstick{\ \vdots \ \ }\\
    \lstick{\ket{\delta_{k, j_n}}}  & \ctrl{1} & \ctrl{2} & \ \ldots\  &\ctrl{4} & \rstick{\ket{\delta_{k, j_n}}}\\
    \lstick[wires=4]{$\ket{0}^{\otimes p}$}  &  \gate{\MOD_{n, p, 0}} & & \ \ldots \ & & \rstick[wires=4]{$E\ket{s_k}$}\\ 
      &  & \gate{\MOD_{n, p, 1}} & \ \ldots \  & &\\
     & & & \ \ldots \ & &\\
     & & & \ \ldots \ & \gate{\MOD_{n, p, p - 1}} &\\
  \end{quantikz}
\end{center}
    For $k \in \{0, \dots p - 1\}$ the above circuit, can be applied in parallel to the appropriate qubits in the encoding; namely those of the form $\ket{\delta_{k, j}}$ for fixed $k$. This leaves us with the state $E\ket{s_0} \otimes \cdots \otimes E\ket{s_{p - 1}}$. Recall that the sum over $\mathbb{F}_p$ we wish to compute is $\sum_{i = 1}^n j_i = \sum_{k = 0}^{p - 1}ks_k$; so, if we can compute each of $ks_k$ and sum over all $k$, we will be left with the desired sum. However, the product $ks_k$ is over elements of $\mathbb{F}_p$ and $p$ is fixed, so it is clear that this can be computed in $\QNC^0$ ($\NC^0$ even). Further, $\sum_{k = 0}^{p - 1}ks_k$ is a sum of constantly many integers each described by $p = O(1)$ bits, which is of course computable by a $\QNC^0$ ($\NC^0$ even) circuit. 
    
    For the sake of completeness, we will describe a circuit composed of permutations on $p$ qubits which compute $\sum_{k = 0}^{p - 1}ks_k$ in our encoded subspace. First let $U_{\sigma}$ be the permutation unitary which satisfies $U_{\sigma}E\ket{j} = E\ket{j - 1 \mod{p}}$. For any $k \in \{0, 1, \dots p - 1\}$, $U_{\sigma}^k$ can be implemented in constant depth via a sequence of at most $p^2$ swap gates. Since $U_{\sigma}^aE\ket{j} = E\ket{j - a}$ for all $a, j \in \Fp$, we can in series apply $U_{\sigma}^{ks_k}$ to $E\ket{b}$ to finally achieve
    \begin{align*}
        \bigg{(}\prod_{k = 0}^{p - 1}U_{\sigma}^{ks_k}\bigg{)}E\ket{b} &= U_{\sigma}^{\sum_{k = 0}^{p - 1}ks_k}E\ket{b}\\
        &= U_{\sigma}^{\sum_{i = 1}^{n}j_i}\ket{b}\\
        &= E\ket{b - \sum_{i = 1}^{n}j_i \mod{p}}
    \end{align*}
    Hence, this gives a circuit of depth $p^3 = O(1)$ and linear size for simulating $\M_{p, n}$ on the encoded qudits. After conjugating by $\Tilde{Q}_p$ gates on the appropriate groups of qubits the equivalence of Lemma \ref{MODA} shows that the entire circuit exactly implements fanout on the encoded qudits. 

    Let's now put all the pieces together to show that we can achieve Task~\ref{step:moore_tricky} of \Cref{cat-to-fanout}. Starting with the state $(\alpha \ket{0} + \beta\ket{1})\ket{0^{(p(n + 1) - 1)}}$, we want to get to an \emph{encoding} of $\alpha \ket{0} + \beta\ket{1}$ and the ancillary qubits. First, apply a $\CNOT$ gate from the first to second qubit, followed by an $X$ gate on the first to obtain the state
    \begin{align*}
        (\alpha\ket{10^{p - 1}} + \beta\ket{010^{p - 2}})\otimes\ket{0^{pn}}.
    \end{align*}
    Now apply an $X$ gate to the $(pj + 1)$st qubit for $j \in \{0, 1, \dots n - 1\}$ yielding the encoded state: $E(\alpha\ket{0} + \beta\ket{1})\otimes E\ket{0^n}$. Note that this is a state on $n + 1$ \emph{qudits} encoded by $p(n + 1)$ \emph{qubits}. After applying the previously described circuit which simulates $\Fanout_{n, p}$ on the encoded states the result is (up to a permutation of the qubits)
    \begin{align*}
        (\alpha\ket{0}^{\otimes 2(n + 1)} + \beta\ket{1}^{\otimes 2(n + 1)})\otimes \ket{0^{(p - 2)(n + 1)}}
    \end{align*}
    Now, via \Cref{cat-to-fanout} any such circuit is sufficient to compute Fanout (on qubits), thus $\QNC^0_{wf} \subseteq \QNC^0[p]$. It is shown in \cite{hs_fanout} and \cite{tt_fanout} that the reverse inclusion holds and it can be concluded that $\QNC^0_{wf} = \QNC^0[p]$.
\end{proof}

\begin{corollary}
    $\QNC^0[a] = \QNC^0_{wf}$ for all $a > 1$.
\end{corollary}
\begin{proof}
    This follows from the previous construction by taking $p$ to be any prime factor of $a$ and setting ancilla qubits appropriately or repeating the input $a/p$ times so that any $\MOD_a$ gate instead computes $\MOD_p$.
\end{proof}

\section{Acknowledgements}
 JM thanks Farzan Byramji for helpful discussions about threshold circuits. Part of this research was performed while the author was visiting the Institute for Mathematical and Statistical Innovation (IMSI), which is supported by the National Science Foundation (Grant No. DMS-1929348).
 
\bibliography{bibli}
\bibliographystyle{plain}
\appendix
\section{Deferred Proofs}\label{PROOFS}
\begin{proof}[Proof of \Cref{cat-to-fanout}]

    To see that (2) $\iff$ (3) it suffices to conjugate either gate by Hadamards i.e. $H^{\otimes(n + 1)}\Fanout_n H^{\otimes(n + 1)} = \Parity_n$. (2) $\implies$ (1) because $F_n$ satisfies the condition described in (1) exactly:
    \begin{align*}
        F_n(\alpha\ket{0} + \beta\ket{1})\ket{0^n} = \alpha\ket{0^n} + \beta\ket{1^n}
    \end{align*}
    Let $C$ be any unitary which satisfies (1). To see that (1) $\implies$ (3) we construct a circuit using $C$ and $C^{\dag}$ in essentially the same way that we did in the proof of \Cref{EQUIV}:
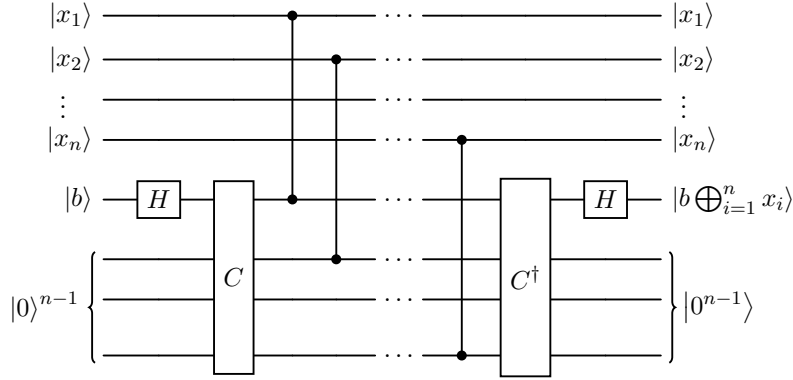
    \begin{center}
    \scalebox{0.9}{
    \begin{quantikz}
    \lstick{\ket{x_1}}  & & & \control{} &  & \ \ldots\  & & & &\rstick{\ket{x_1}}\\
    \lstick{\ket{x_2}}  & & & & \control{} & \ \ldots\  & & & & \rstick{\ket{x_2}}\\
    \lstick{\vdots \ \ \ } & & & & &\ \ldots \  & & & & \rstick{\ \vdots \ \ }\\
    \lstick{\ket{x_n}}  & & & & &\ \ldots\  & \control{} & & & \rstick{\ket{x_n}}\\
    \lstick{\ket{b}}  & \gate{H} &\gate[wires=4]{C} & \ctrl{-4} & & \ \ldots\  &  & \gate[wires=4]{C^{\dag}} & \gate{H} & \rstick{\ket{b \bigoplus_{i = 1}^n x_i}}\\
    \lstick[wires=3]{$\ket{0}^{n - 1}$} & & & & \ctrl{-4} & \ \ldots \ & & & & \rstick[wires=3]{$\ket{0^{n - 1}}$}\\ 
    & & &  &  & \ \ldots \  & & & & \\
     & & & & & \ \ldots \ & \ctrl{-4} & & & 
  \end{quantikz}
  }
\end{center}
To see that this circuit exactly computes parity note that after the first Hadamard gate is applied the ancilla bits are in the state $\frac{\ket{0} + (-1)^b\ket{1}}{\sqrt{2}}\ket{0^{n - 1}}$ and after applying $C$ we have $\frac{\ket{0^n} + (-1)^b\ket{1^n}}{\sqrt{2}}$. After the $\CZ$ gates are applied we are left with the state $\frac{\ket{0^n} + (-1)^{b + \sum_{i = 1}^n x_i}\ket{1^n}}{\sqrt{2}}$. Since $C(\alpha\ket{0} + \beta\ket{1})$After applying $C^{\dag}$ we are left with $\frac{\ket{0} + (-1)^{b + \sum_{i = 1}^n x_i}\ket{1}}{\sqrt{2}}\ket{0^{n - 1}}$, so the final Hadamard gate leaves the output qubit and the ancilla qubits in the state $\ket{b\bigoplus_{i = 1}^nx_i}\ket{0^{n - 1}}$.
\end{proof}
It should be noted that when our circuit $C$ has property $(1)$ it is in some sense stronger than the guarantee $C\ket{0^n} = \frac{\ket{0^n} + \ket{1^n}}{2}$. In the case of the latter it seems that an $\AND$ gate is required to compute parity with $C$ and $C^{\dag}$, which \Cref{cat-to-fanout} shows is not necessary when $C(\alpha\ket{0} + \beta\ket{1})\ket{0^{n - 1}} = \alpha\ket{0^n} + \beta\ket{1^n}$ for all one-qubit states $\alpha\ket{0} + \beta\ket{1}$.

To see that (2) $\implies$ (1) note that
$$
\Fanout_n \ket{+} \ket{0^n} = \frac{\ket{0^{n + 1}} + \ket{1^{n + 1}}}{\sqrt{2}}
$$
 Observe that if given access to some circuit $C$ (and $C^{\dag}$) which satisfies the weaker condition of preparing an exact $n$-nekomata i.e \begin{align*}C\ket{0^m} = \frac{\ket{0^n}\ket{\psi_0} + \ket{1^n}\ket{\psi_1}}{\sqrt{2}}\end{align*} one can construct a constant-depth $\QAC$ circuit which exactly computes Parity:
\begin{figure}
    \centering
    \input{figures/nekomata_to_parity_circuit}
    \caption{Computing $\Parity_n$ with a circuit $C$ which prepares an exact $n$-nekomata and its inverse $C^{\dag}$}
    \label{fig:nekomata_to_parity}
\end{figure}
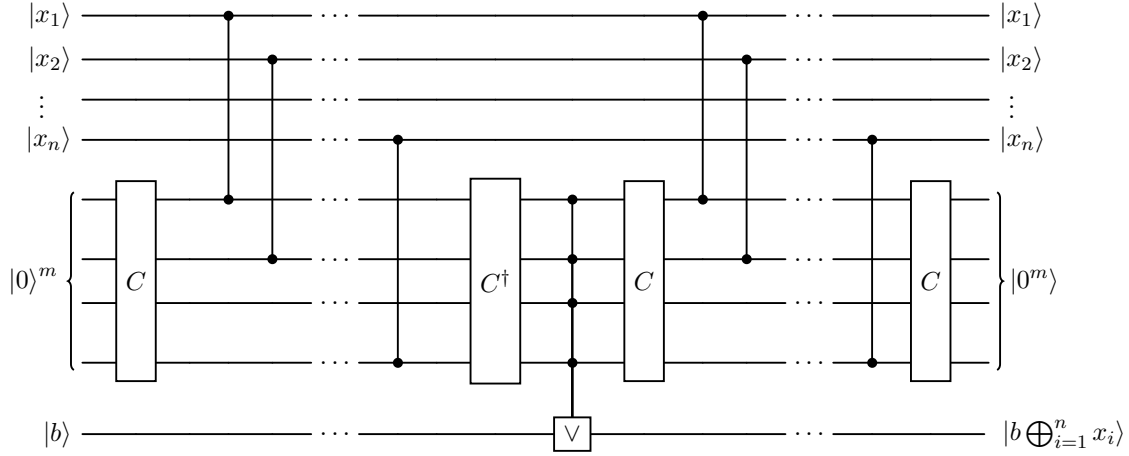

Note that in the above circuit after the first layer of $\CZ$ gates are applied the state on the ancialla qubits is $\frac{\ket{0^n}\ket{\psi_0} + (-1)^{|x|}\ket{1^n}\ket{\psi_1}}{\sqrt{2}}$. When $|x|$ is even, nothing has happened, so $C^{\dag}$ will return the state on these registers to $\ket{0^m}$ and the $\lor$-gate will not change the final register. When $|x|$ is odd $\frac{\ket{0^n}\ket{\psi_0} + (-1)^{|x|}\ket{1^n}\ket{\psi_1}}{\sqrt{2}}$ is orthogonal to $C\ket{0^m}$, so after applying $C^{\dag}$ the resulting state is orthogonal to $\ket{0^m}$ which will always trigger the $\lor$-gate. The second half of the circuit uncomputes returning the ancillary registers to $\ket{0^m}$. Thus, the final register is always left in the state $\ket{b \bigoplus_{i = 1}^nx_i}$ - thus, (up to an $X$ gate) this circuit exactly computes parity.

Finally, we prove \Cref{approx-NEK-PAR-FAN} and in particular that the above circuit approximates $\Parity_n$ when $C$ is replaced with any $U$ which produces an $\epsilon$-approximate $n$-nekomata when applied to the all zeros state.
\begin{proof}[Proof of \Cref{approx-NEK-PAR-FAN}]
    Let $U$ be a unitary on $m$ qubits such that $U\ket{0^m} = \ket{\psi}$ is an $\epsilon$-approximate $n$-nekomata and let $\ket{\nu} = \frac{\ket{0^n}\ket{\psi_0} + \ket{1^n}\ket{\psi_1}}{\sqrt{2}}$ be the $n$-nekomata on $m$ qubits which maximizes $|\braket{\nu}{\psi}|^2$. We can write $\ket{\psi} = \sqrt{1 - \epsilon}\ket{\nu} + \sqrt{\epsilon}\ket{\nu^{\perp}}$ for some $\ket{\nu^{\perp}}$ which is orthogonal to $\ket{\nu}$. In the circuit shown in \Cref{fig:nekomata_to_parity} observe that after the first layer of $\CZ$ gates are applied the state on the ancilla registers is 
    $$
    \ket{\psi_{-}}\sqrt{1 - \epsilon}\frac{\ket{0^n}\ket{\psi_0} + (-1)^{|x|}\ket{1^n}\ket{\psi_1}}{\sqrt{2}} + \sqrt{\epsilon}\ket{\nu_{-}^{\perp}}
    $$
    For some $\ket{\nu^{\perp}_{-}}$. Observe that when $|x|$ is even then $|\braket{\psi_{-}}{\psi}|^2 \geq 1 - 2\epsilon$ and when $|x|$ is odd then $|\braket{\psi_{-}}{\psi}|^2 \leq \epsilon$. In the former case $C^{\dag}\ket{\psi_{-}}$ will have fidelity at least $1 - 2\epsilon$ with $\ket{0^m}$, so after uncomputation we are left with $\ket{x, 0^m}\ket{\omega_b}$ where $|\braket{b}{\omega_b}|^2 \geq 1 - 2\epsilon$. Similarly, when $|x|$ is odd $C^{\dag}\ket{\psi_{-}}$ has fidelity at most $\epsilon$ with $\ket{0^m}$ and after uncomputation we are left with $\ket{x, 0^m}\ket{\omega_b}$ where  $|\braket{b}{\omega_b}|^2 \leq \epsilon$. Thus, on any input $\ket{x, b}$ the state produced by this circuit has fidelity at least $1 - 2\epsilon$ with $F_n\ket{x, b}$ - equivalently, the $\ell_2$-distance is at most $2\epsilon$ meaning that the unitary implemented by this circuit, $V$, satisfies $\|V  - \Parity_n\|_{\mathrm{op}}\leq 2\epsilon$. Thus, (1) $\implies$ (3).

    To see that (3) $\implies$ (2) suppose that the unitary $U$ satisfies $\|U - \Parity_n\|_{\mathrm{op}} \leq \epsilon$. Then,
    \begin{align*}
        \|H^{\otimes(n + 1)}UH^{\otimes(n + 1)} - \Fanout_n\|_{\mathrm{op}} &= \|H^{\otimes(n + 1)}(U - \Parity_n)H^{\otimes(n + 1)}\|_{\mathrm{op}} \leq \|U - \Parity_n\|_{\mathrm{op}} \leq \epsilon
    \end{align*}
    For (2) $\implies$ (1) let $\ket\psi = \Fanout_n\ket{+}\ket{0^n} = \frac{\ket{0^n} + \ket{1^n}}{\sqrt{2}}$ and $\ket\phi = U\ket{+}\ket{0^n}$. Note that if $\|U - \Fanout_n\|_{\mathrm{op}} \leq \epsilon$ then
    \begin{align*}
        \|\ket\phi - \ket{\psi}\|_2 \leq \|\ket{+}\ket{0^{n}}\|_2 \|U - \Fanout_n\|_{\mathrm{op}} \leq \epsilon
    \end{align*}
    So, 
    \begin{align*}
        \|\ket{\phi} - \ket{\psi}\|_2 &= \sqrt{2 - \braket{\psi}{\phi} - \braket{\phi}{\psi}} \leq \epsilon \implies |\braket{\psi}{\phi}|^2 \geq 1 - \epsilon^2 - \epsilon^4/4
    \end{align*}
    Thus, $\ket{\phi}$ is an $O(\epsilon)$-approximate nekomata.
\end{proof}

\begin{proof}{Proof of \Cref{MODA}}
    This can be seen via a direct computation:  \begin{align*}
       \Fanout_{n, p}Q^{\otimes{n + 1}}\ket b \ket x &=  \Fanout_{n, p} \bigg{(}\frac{1}{\sqrt{p}}\sum_{j = 0}^{p - 1} \omega^{bj}\ket{j} \bigg{)} \otimes \bigg{(}\frac{1}{\sqrt{p^n}}\sum_{y \in \Fp^n} \omega^{\langle x, y \rangle}\ket{y} \bigg{)}
    \end{align*}
    Here $\langle x , y \rangle$ denotes the inner product over vectors in $\Fp^n$: $\langle x , y \rangle = \sum_{j = 1}^p x_jy_j \mod{p}$. For $y \in \Fp^n$ and $j \in \Fp$ we will use $y^{(j)}$ to denote the string obtained by adding $j$ to every entry of $y$ i.e. $F_p\ket{j}\ket y = \ket j\ket{y^{(j)}}$. Now we can see that
    \begin{align*}
         \bigg{(}\frac{1}{\sqrt{p}}\sum_{j = 0}^{p - 1} \omega^{bj}\ket{j} \bigg{)} \otimes \bigg{(}\frac{1}{\sqrt{p^n}}\sum_{y \in \Fp^n} \omega^{\langle x, y \rangle}\ket{y} \bigg{)} = \frac{1}{\sqrt{p^{n + 1}}}\sum_{y \in \Fp^n}\sum_{j = 0}^{p - 1} \omega^{bj + \langle x, y\rangle} \ket{j}\ket{y}
    \end{align*}
    So, 
    \begin{align*}
        \Fanout_{n, p} \frac{1}{\sqrt{p^{n + 1}}}\sum_{y \in \Fp^n}\sum_{j = 0}^{p - 1} \omega^{bj + \langle x, y\rangle} \ket{j}\ket{y} &=  \frac{1}{\sqrt{p^{n + 1}}}\sum_{y \in \Fp^n}\sum_{j = 0}^{p - 1} \omega^{bj + \langle x, y\rangle} \ket{j}\ket{y^{(j)}}
    \end{align*}
    After rearranging we have
    \begin{align*}
        \frac{1}{\sqrt{p^{n + 1}}}\sum_{y \in \Fp^n}\sum_{j = 0}^{p - 1} \omega^{bj + \langle x, y\rangle} \ket{j}\ket{y^{(j)}} &= \frac{1}{\sqrt{p^{n + 1}}}\sum_{y \in \Fp^n} \sum_{j = 0}^{p - 1} \omega^{bj +\langle x, y^{(-j)}\rangle}\ket{j}\ket{y}\\
        &= \frac{1}{\sqrt{p^{n + 1}}}\sum_{y \in \Fp^n} \sum_{j = 0}^{p - 1} \omega^{bj +\langle x, y\rangle - j|x|}\ket{j}\ket{y}\\
        &= \frac{1}{\sqrt{p^{n + 1}}}\sum_{y \in \Fp^n} \sum_{j = 0}^{p - 1} \omega^{\langle x, y\rangle}\omega^{j(b - |x|)}\ket{j}\ket{y}\\
        &= Q^{\otimes (n + 1)}\M_{n, p}\ket b\ket{x}
    \end{align*}
    Thus, $\M_{n, p} = (Q_p^\dag)^{\otimes{n + 1}}\Fanout_{n, p}Q_p^{\otimes{n + 1}}$ as claimed. 
\end{proof}
\end{document}

%% file: figures/nekomata_construction.tex
\newcommand{\xmax}{4}
\newcommand{\ymax}{4}
\newcommand{\tikzscale}{.8}
\newcommand{\xpicshift}{3}
\newcommand{\stageoffset}{.5}

\begin{tikzpicture}[scale=\tikzscale]
  \begin{scope}
  \foreach \y in {1,...,\ymax} {
    \node[circle, fill, color=blue, inner sep=1.5pt, label={north east:{\tiny $\ket{0}$}}] (T\y) at (-.3,\y) {};
  }

  \foreach \x in {1,...,\xmax} {
      \foreach \y in {1,...,\ymax} {
          \node[circle, fill, inner sep=1.5pt, label={north east:{\tiny $\ket{1}$}}] (N\x\y) at (\x,\y) {};
          
      }
  }

  \draw[decorate,decoration={brace,amplitude=10pt}] 
  ($(T1)+(-.2,-.2)$) --  node[midway,left=10pt]{$n$} ($(T\ymax)+(-.2,.2)$);
  \draw[decorate,decoration={brace,amplitude=10pt}] 
  ($(N1\ymax)+(-.2,.5)$) --  node[midway,above=10pt]{$m$} ($(N\xmax\ymax)+(.2,.5)$);
\end{scope}

\begin{scope}[shift={(\xmax + \xpicshift,0)}]
  \foreach \y in {1,...,\ymax} {
    \node[circle, fill, inner sep=1.5pt, color=blue] (T\y) at (-.3,\y) {};
  }

  \foreach \x in {1,...,\xmax} {
      \foreach \y in {1,...,\ymax} {
          \node[circle, fill, inner sep=1.5pt] (N\x\y) at (\x,\y) {};
          
      }
  }

  \foreach \x in {1,...,\xmax} {
      \draw[rounded corners=5pt] (\x-0.2, .8) rectangle (\x+0.2, \ymax+.2);
}
\end{scope}
\begin{scope}[shift={(2*\xmax + 2*\xpicshift,0)}]
  \foreach \y in {1,...,\ymax} {
    \node[circle, fill, color=blue, inner sep=1.5pt] (T\y) at (-.3,\y) {};
  }

  \foreach \x in {1,...,\xmax} {
      \foreach \y in {1,...,\ymax} {
          \node[circle, fill, inner sep=1.5pt] (N\x\y) at (\x,\y) {};
          
      }
  }

  \foreach \y in {1,...,\xmax} {
      \draw[rounded corners=5pt] (-0.5, \y-.2) rectangle (\xmax+0.2, \y+.2);
}
\end{scope}
\node at (\xmax/2,-\stageoffset) {\shortstack{Stage 1: \\ Initialization}};
\node at (\xmax/2+\xmax+\xpicshift,-\stageoffset) {\shortstack{Stage 2: \\ $R_{\psi_S}$ gates}};
\node at (\xmax/2+2*\xmax+2*\xpicshift,-\stageoffset) {\shortstack{Stage 3: \\ Toffoli gates}};
\end{tikzpicture}

%% file: figures/nekomata_to_parity_circuit.tex
\scalebox{0.9}{
    \begin{quantikz}
    \lstick{\ket{x_1}}  & & & \control{} &  & \ \ldots\  & & & & & & \control{} & & \ \ldots \ & & & \rstick{\ket{x_1}}\\
    \lstick{\ket{x_2}} & & & & \control{} & \ \ldots\  & & & & & & & \control{} & \ \ldots \ & & &  \rstick{\ket{x_2}}\\
    \lstick{\vdots \ \ \ } & & & & & \ \ldots \  & & & & & & & & \ \ldots \ & & &\rstick{\ \vdots \ \ }\\
    \lstick{\ket{x_n}}  & & & & & \ \ldots\  & \control{} & & & & & & & \ \ldots \ & \control{} & &\rstick{\ket{x_n}}\\
    \lstick[wires=4]{$\ket{0}^m$} & \gate[wires=4]{C} & &  \ctrl{-4} & & \ \ldots \ & & & \gate[wires=4]{C^{\dag}} &\ctrl{4} &\gate[wires=4]{C} &\ctrl{-4} & & \ \ldots \ & & \gate[wires=4]{C} & \rstick[wires=4]{$\ket{0^m}$}\\ 
    & & &  & \ctrl{-4} & \ \ldots \  & & & & \ctrl{3} & & & \ctrl{-4} & \ \ldots \ & & &\\
     & & & & & \ \ldots \ & & & & \ctrl{2} & & & & \ \ldots \ & & &\\
     & & & & & \ \ldots \ & \ctrl{-4} & & & \ctrl{1}& & & & \ \ldots \ & \ctrl{-4} & &\\
     \lstick{\ket{b}}&  & & & & \ \ldots \ & & & & \gate{\lor} & & & & \ \ldots \ & & & \rstick{\ket{b\bigoplus_{i = 1}^nx_i}}\
  \end{quantikz}
  }